\documentclass[12pt]{article}
\usepackage[mathscr]{eucal}
\usepackage{amsmath,amsfonts,amssymb,amsthm,mathtools}
\usepackage{times}
\usepackage{hyperref}

\advance\textheight\topskip
\numberwithin{equation}{section} \makeatletter
\@addtoreset{equation}{section}

\newtheorem{prop}{Proposition}[section]
\newtheorem{lemma}[prop]{Lemma}

\renewcommand{\tilde}{\widetilde}
\renewcommand{\hat}{\widehat}

\newcommand{\bref}[1]{\textbf{\ref{#1}}}

\newcommand{\dd}{\partial}
\renewcommand{\d}{\partial}

\renewcommand{\geq}{\,{\geqslant}\,}
\renewcommand{\leq}{\,{\leqslant}\,}

\newcommand{\binner}[2]{%
  {\langle}\kern-4.15pt{\langle}#1{,}\,#2{\rangle}\kern-4.15pt{\rangle}}

\newcommand{\half}{\mathchoice{%
    \ffrac{1}{2}}{\frac{1}{2}}{\frac{1}{2}}{\frac{1}{2}}}

\newcommand{\ffrac}[2]{\raisebox{.5pt}%
  {\footnotesize$\displaystyle\frac{#1}{#2}$}\kern1pt}

\newcommand{\dl}[1]{\mathchoice{\ffrac{\dd}{\dd #1}}{\frac{\dd}{\dd
      #1}}{\ffrac{\dd}{\dd #1}}{\ffrac{\dd}{\dd #1}}}

\newcommand{\fR}{\mathbb{R}}

\def\cD{\mathcal{D}}
\def\cE{\mathcal{E}}

\def\cV{\mathcal{V}}

\def\BG-Poincare{Barnich:2009jy}
\def\Fedosov-book{Fedosov:1996fu}

\usepackage{graphicx, wrapfig, subcaption, setspace, booktabs}

\begin{document}
\begin{center}
{\Large\bfseries 
On the structure of the conformal higher-spin wave operators
\vspace{0.4cm}
} \\

\vskip 0.04\textheight

Maxim Grigoriev${}^{a}$ and Aliaksandr Hancharuk ${}^{a,b}$
%<goncharuk@phystech.edu>
\vskip 0.04\textheight

\vspace{5pt}
{\em$^a$ Tamm Department of Theoretical Physics, \\Lebedev Physics
  Institute,\\ Leninsky ave. 53, 119991 Moscow, Russia}\\

\vspace{5pt}
{\em$^b$ Department of General and Applied Physics,\\ 
Moscow Institute of Physics and Technology, \\
Institutskiy per. 7, Dolgoprudnyi, 141700 Moscow region, Russia}\\

\vskip 0.02\textheight

\vspace{2cm}

{\bf Abstract }

\end{center}
\begin{quotation}
We study conformal higher spin (CHS) fields on constant curvature backgrounds. By employing parent formulation technique in combination with tractor description of GJMS operators we find a manifestly factorized form of the CHS wave operators for symmetric fields of arbitrary integer spin $s$ and gauge invariance of arbitrary order $t\leq s$. In the case of the usual Fradkin-Tseytlin fields $t=1$ this gives a systematic derivation of the factorization formulas known in the literature while for $t>1$ the explicit formulas were not known. We also relate the gauge invariance of the CHS fields to the partially-fixed gauge invariance of the factors and show that the factors can be identified with (partially gauge-fixed) wave operators for (partially)-massless or special massive fields. As a byproduct, we establish a detailed relationship with the tractor approach and, in particular, derive the tractor form of the CHS equations and gauge symmetries.
\end{quotation}

\thispagestyle{empty}
\newpage
\cleardoublepage
\setcounter{page}{1}

\tableofcontents
\section{Introduction}

Conformal higher spin gauge theories attract considerable attention because despite being non-unitary they give tractable examples of interacting Lagrangian theories extending conformal gravity and involving higher spin fields. The simplest conformal higher spin (CHS) fields are totally symmetric tensor fields subject to first order gauge transformations. These are also known as Fradkin-Tseytlin fields and were originally proposed in~\cite{Fradkin:1985am} in 4 dimensions and generalized in~\cite{Segal:2002gd} to higher even dimensions. Interacting theory for these fields was proposed much later~\cite{Segal:2002gd,Tseytlin:2002gz} and elaborated further in~\cite{Bekaert:2010ky} (see also~\cite{Bonezzi:2017mwr} for a recent discussion).

CHS fields are intimately related to Fronsdal fields in anti de Sitter space (AdS) of one extra dimension in the context of AdS/CFT correspondence. More specifically, CHS fields in $n$-dimensions can be regarded as leading boundary values of the Fronsdal fields in AdS space of ${n+1}$ dimensions. In so doing the CHS Lagrangian arises as the holographic Weyl anomaly \cite{Metsaev:2008fs,Metsaev:2009ym} (see also~\cite{Bekaert:2012vt,Bekaert:2013zya} for the gauge-covariant analysis at the level of equations of motion). 

Wave operators for CHS fields (CHS operators) are of order $d-4+2s$ and were conjectured~\cite{Tseytlin:2013jya} to factorize into a product of 2nd order operators when written  over the constant curvature background (see also~\cite{Metsaev:2007fq,Metsaev:2007rw,Joung:2012qy} for the relevant earlier contributions). It turns out that in 4 dimensions the factors have the mass terms identical to those of partially-massless fields~\cite{Deser:1983mm,Deser:2001pe} whose order of gauge transformation (known as ``depth'') ranges from $1$ to $s$. In higher dimensions in addition to partially-massless-like wave operators one also finds among the factors the wave operators of certain massive fields~\cite{Metsaev:2014iwa}. The existence of the factorized form allows to express the partition function of CHS fields in terms of the known partition functions of the partially-massless fields~\cite{Tseytlin:2013jya} (see also~\cite{Beccaria:2015vaa,Beccaria:2016tqy}). The manifestly factorized form of the CHS operators was given in~\cite{Nutma:2014pua}.

For a scalar ($s=0$) conformal field the factorization amounts to the familiar factorization of higher-order conformal operators known in the context of conformal geometry as GJMS operators~\cite{Paneitz:1983,Fradkin:1981jc,GJMS}. These also have a natural generalization~\cite{Gover:2005mn} to tractor fields on conformally-Einstein manifolds (for an introduction to tractors see e.g.~\cite{Eastwood,BEG,Cap:2002aj}). The similarity with GJMS operators suggests that tractor technique can be useful in studying factorization of CHS wave operators as well. 

As far as the structure of CHS operators on constant curvature background is concerned in addition to manifestly-factorized form~\cite{Nutma:2014pua} it is also worth mentioning a suggestive ordinary derivative Lagrangian formulations proposed in~\cite{Metsaev:2014iwa}. Furthermore, the manifestly-conformal formulation of CHS equations was proposed in~\cite{Bekaert:2012vt,Bekaert:2013zya} by employing a version of the ambient space technique. In our study of the CHS operators we use this formulation as a starting point.

In this work we are concerned with more general class of the CHS wave operators for totally symmetric fields, which includes those with gauge invariance of arbitrary order $t\leq s$.  These operators are of order $d-2+2(s-t)$ and can be considered as those of the boundary values of depth-$t$ partially-massless fields in $AdS_{n+1}$. We relate these CHS wave operators to GJMS ones by constructing a certain embedding of tensor fields into tractors. More specifically, we realize tractors  using the parent-formulation approach~\cite{Barnich:2006pc,Bekaert:2009fg}, where the usual ambient-space construction is employed to describe the tangent space rather than the spacetime, and demonstrate that CHS fields can be embedded in such a way that GJMS operators coincides with CHS wave operators. As a byproduct we clarify and elaborate in some details on the parent approach description of tractors originally put forward in~\cite{Grigoriev:2011gp}.

Constructing the factorized form requires introducing the so-called scale tractor~\cite{Gover:2005mn} and a new ingredient -- operators $B_k$, $k=1,2,\ldots$. For $k=1$ this was already employed in the literature in the context of tractor description of low spin~\cite{Gover:2008sw} and higher-spin~\cite{Grigoriev:2011gp} fields in constant curvature spaces. In the formulation developped in this paper CHS operator simply takes a manifestly factorized form $B_\ell \ldots B_1$, $\ell=\frac{d-4}{2}+s$. This representation turns out to be useful in analyzing gauge invariance. In particular, we show that $B_t$ is the wave operator of depth-$t$ partially-massless field in traceless gauge.

The paper is organized as follows: in Section~\bref{sec:ambient} we recall the ambient space formulation of tractors and GJMS operators. There we also introduce parent formulation technique which allows to work with fields defined on the conformal space rather than ambient one but still benefit from the manifest realizations of $o(n,2)$-symmetry. In section~\bref{sec:CHS}
we introduce CHS fields, propose a new manifestly $o(n,2)$-invariant formulation for them, and construct the manifestly factorized factorized form. Technical details are relegated to Appendices.

\section{Ambient space and tractors in the parent approach}
\label{sec:ambient}
\subsection{Ambient space}

In this work we are concerned with conformal gauge fields defined on the conformally-flat spaces. The conformal symmetry can be seen as originating from the conformal isometries which, at the infinitesimal level, are given by conformal Killing vector fields. These form $o(n,2)$ algebra, where $n$ is the space-time dimension.

Conformally invariant equations can be described~\cite{Dirac:1936fq} in a manifestly $o(n,2)$-invariant way by employing the ambient space construction (which in turn originates from that of Klein). An ambient space is a pseudo-Eucledean space $\mathbb{R}^{n,2}$ equipped with the metric $\eta_{AB}$ of signature $(-,+,+...,+,-)$. In what follows we use ambient coordinates $(X^+, X^a, X^-)$ so the metric has the form
\begin{equation}
	\eta_{AB}dX^AdX^B = 2dX^+dX^- + \eta_{ab}dX^adX^b
\end{equation}
Where $\eta_{ab}$ is Minkowski metric $(n-1, 1)$. The cone is a zero locus $\lbrace X^2 = 0 \rbrace \backslash \lbrace 0\rbrace$. In this picture, the $n$-dimensional conformal space $M$ is the projectivization of the cone (projective cone in what follows), i.e. the quotient space of $X^2=0$ modulo the equivalence relation $X^A \sim \lambda  X^A, \lambda \in \fR \backslash \lbrace 0 \rbrace $. The quotient is equipped with the conformal structure and with a natural action of $o(n,2)$ as well as the entire conformal group (in what follows we restrict to infinitesimal analysis and hence concentrate on the conformal algebra). The action comes from the standard $o(n,2)$-action on the ambient space.

To pick a representative of the equivalence class of metrics on $M$ one can embed $M$ as a submanifold of $X^2=0$ such that each ray intersects $M$ once and only once. The metric (which is conformally flat by construction) is then obtained by pulling back the ambient metric to $M$. Scalar densities of conformal weight $w$ on $M$ can be described ambiently as:
\begin{subequations}
\label{amb-scalar}
\begin{gather}
		(X\cdot\frac{\partial}{\partial X} - w)\Phi = 0\\
		\Phi \sim \Phi + X^2\chi\,,
	\end{gather}
\end{subequations}
in terms of the ambient space functions $\Phi=\Phi(X)$. Here and in what follows $\cdot$ denotes $o(n,2)$-invariant contraction of indices, e.g. $Z\cdot W=\eta_{AB}Z^A W^B$ and $X^2=X\cdot X$. Because both the constraint and the equivalence relation are manifestly $o(n,2)$-invariant the conformal algebra act on the space defined by~\eqref{amb-scalar} (the same applies to the conformal group).

In a similar fashion we can consider tensor fields on the ambient space satisfying an analog of~\eqref{amb-scalar}. If we restrict ourselves to  totally symmetric fields it is convenient to work in terms of generating functions defined on the cotangent bundle over the ambient space
\begin{equation}
\Phi(X,P)=\sum_{i=0} \Phi^{A_1 \ldots A_i}P_{A_1}\ldots P_{A_i}\,.
\end{equation}
Here $P_A$ are coordinates on the fibers and we assume $\Phi$ to be polynomial in $P$. It follows that the space defined by
\begin{subequations}
\label{ambient_tractors}
\begin{gather}
		(X\cdot \frac{\partial}{\partial X} - w)\Phi(X,P) = 0 \label{at1}\\
		\Phi(X) \sim \Phi(X) + X^2\chi(X,P) \label{at2}
	\end{gather}
\end{subequations}
is that of totally symmetric tractors of weight $w$, which we denote by $\mathcal{E}^{\bullet}[w]$. Indeed, equation \eqref{at2} implies that such tensors are actually defined on the cone $X^2=0$ while \eqref{at1} says that these are actually defined on the projective cone. It is clear that $\mathcal{E}^{\bullet}[w]$ is equipped with a natural action of $o(n,2)$ induced by that on the cotangent bundle over the ambient space.

\subsection{GJMS operators}

Representing tractor fields through~\eqref{ambient_tractors} is useful in studying conformally invariant differential operators defined on tractors. In particular, in these terms it is easy to define so-called GJMS-operators. These were originally proposed~\cite{GJMS} for scalar densities and later extended to generic tractors~\cite{Gover:2002ay,Gover:2005mn} using the Feffermann-Graham construction which reduces to the above ambient space approach in the conformally flat case.

If tractors are described through~\eqref{ambient_tractors} the GJMS-operators are simply powers of the ambient Laplacian
\begin{equation}
	P^{2\ell} \coloneqq \Box_X^{\ell},\quad  P^{2\ell}: \mathcal{E}^{\bullet}[\ell-\frac{n}{2}] \mapsto \mathcal{E}^{\bullet}[-\ell-\frac{n}{2}]\,, \qquad \Box_X\coloneqq \dl{X}\cdot\dl{X}
\end{equation}
This operator is well defined on equivalence classes~\eqref{at2} provided the weights are as above. To see this, it is instructive to exploit that the following 3 operators
\begin{equation}
H \coloneqq X\cdot\dl{X} +\frac{n+2}{2}\,, \qquad E\coloneqq -\half X^2\,, \qquad F \coloneqq  \half \dl{X}\cdot \dl{X}\,.
\end{equation}
define a representation of $sl(2)$-algebra on the ambient space functions. For trivial $\Phi(X)=X^2 \chi$ with $\chi$ of weight $w = \ell-2-\frac{n}{2}$ one finds 
\begin{equation}
	\Box_X^{\ell}(X^2\chi) = X^2\Box_X^{\ell}\chi + 4\ell\Box_X^{\ell-1}(w_{\chi} + \frac{n}{2} - \ell +2)\chi = X^2\Box_X^{\ell}\chi
\end{equation}
because $w_{\chi} = w_{\Phi} - 2 = \ell - 2 - \frac{n}{2}$.

To write explicit formulas for GJMS operators one chooses local coordinates $x^\mu$ on $M$ and particular metric $g_{\mu\nu}$ in the conformal class. The first non-trivial example of GJMS operators is Yamabe operator $\bar\nabla^2 - \frac{(n- 2)}{4(n-1)}R$ defined on scalar densities of weight $1- \frac{n}{2}$, where $\bar\nabla$ is the Levi-Civita connection determined by $g$. 
If one identifies $M$ with flat Minkowski space then GJMS operators are just $\bar\nabla^{2\ell}$ acting on scalar densities of weight $w = \ell - \frac{n}{2}$.

There is a dual description of GJMS operators 
using the following system \cite{GJMS} (see also~\cite{Bekaert:2013zya}):
\begin{equation}
\label{dualGJMS}
	\begin{gathered}
		\Box_X \Phi(X) = 0\,,\\
		(X\cdot\dl{X}-\ell+\frac{n}{2})\Phi(X)=0\,,\\		
		\Phi(X) \sim \Phi(X) + (X^2)^{\ell}\chi\,.
	\end{gathered}
\end{equation}
Here $\chi$ also satisfies analogous constraints but with $-\ell+\frac{n}{2}$ replaced with $\ell+\frac{n}{2}$.
It follows this system is equivalent to~\eqref{ambient_tractors} supplemented by $\Box_X^{\ell} \Phi=0$ and with $w=\ell-\frac{n}{2}$. It is a remarkable property of~\eqref{dualGJMS} that the first two equations considered in the vicinity of the hyperboloid
$X^2=-1$ in the ambient space describe the scalar field of mass $w(w+n)$. Moreover, this gives a systematic way~\cite{Bekaert:2012vt,Bekaert:2013zya} to describe boundary values of the (A)dS field on the hyperboloid. 

\subsection{Thomas-D operator}
An important object well-defined on the equivalence classes \ref{ambient_tractors} is Thomas D-operator $D^A:\mathcal{E}^{B..C}[w] \mapsto \mathcal{E}^{AB..C}[w-1]$. Here we denote by $\mathcal{E}^{B..C}[w]$ tensor fields of arbitrary symmetry and rank of weight $w$ satisfying:
\begin{subequations}
\label{ambient_tractors2}
\begin{gather}
		(X\cdot \frac{\partial}{\partial X} - w)\Phi = 0\\
		\Phi \sim \Phi + X^2\,\chi \
	\end{gather}
\end{subequations}
Of course, \eqref{ambient_tractors2} contains \eqref{ambient_tractors}.
By slightly abusing notations, Thomas D-operator can be defined as follows
\begin{equation}
\label{Thomas}
	D_A \Phi = (2(X \cdot \dl{X} + \frac{n}{2})\frac{\partial}{\partial X^A} - X_A\Box_X)\Phi\,,
\end{equation}
where $\Phi$ is subject to~\ref{ambient_tractors2}. 
It has the following properties:
\begin{itemize}
 \item  $[D_A, D_B] = 0$ 
 \item  $D_A D^A = X^2\Box_X^2$\,.
\end{itemize}
Note that the above definition and properties are to be modified in the conformally non-flat case.

There is a useful relation between GJMS operators and Thomas-D observed in~\cite{Gover:2002ay}:
\begin{equation}
\label{D->P1}
	D_{A_1}...D_{A_{\ell}}\Phi = (-1)^kX_{A_1}...X_{A_{\ell}}P^{2\ell}\Phi\,, \qquad \Phi \in \mathcal{E}^{\bullet}[\ell-\frac{n}{2}]
\end{equation}
A version of this relation is going to be very useful in what follows.

\subsection{Tractors in parent formulation}\label{sec:tractors-parent}

The naive ambient construction sketched above is very useful in describing fields on a conformally-flat background. However, it operates in terms of the equivalence classes of fields on the ambient space rather then fields explicitly defined on the conformal manifold. Moreover, ambient description is not directly applicable to a conformally-flat space which is equivalent to the projective cone only locally. 

Although these issues can be resolved by employing a full-scale Fefferman-Graham construction~\cite{FG} there is a relatively simple and concise alternative. It is based on
reformulating the system~\eqref{ambient_tractors} in the so-called parent form~\cite{Barnich:2006pc,Bekaert:2009fg} in which the ambient construction is realized in the formal version of the ambient space rather than in the space-time. Moreover, this approach has proved useful in describing gauge fields including CHS fields and hence provides a framework to study the structure of the CHS wave operators.

The parent counterpart of the system~\eqref{ambient_tractors} is constructed by first introducing the formal version of the ambient space with coordinates $Y^A$, where one considers totally symmetric tensor fields. As before we work with tensors in terms of the generating function
\begin{equation}
 \Phi=\Phi(Y,P)\,.
\end{equation} 
The dependence on $Y$ is assumed formal, i.e. as functions on the formal ambient space one takes
polynomials in $P$ with coefficients in formal series in $Y$. Given a nonvanishing ambient vector $V_0^A$
one defines a ``twisted'' realization~\cite{Barnich:2006pc,Bekaert:2009fg} of $o(n,2)$ on the space of the above functions in $Y,P$:
\begin{equation}
\label{twisted}
\rho(\alpha)\Phi=\alpha^A_B\left[P_A\dl{P_B}-(Y^B+V_0^B)\dl{Y^A}\right]\,, \qquad \alpha \in o(n,2)\,.
\end{equation}

Then, the conformal structure on $M$ can be encoded in terms of the vector bundle $\cV$ over $M$ whose fiber is a copy of the flat ambient space. More precisely, the bundle is equipped with the fiber-wise pseudo-Euclidean metric $\eta$, nonvanishing section $V\in \Gamma (\cV)$ such that $\eta(V,V)=0$, and an $o(d,2)$ connection $dx^\mu\omega_\mu$ compatible with $\eta$ and such that $\nabla V$ has maximal rank (i.e. seen as a fiber-wise map $TM\to\cV(M)$ it has a vanishing kernel). As we are now interested in flat conformal structures we restrict ourselves to flat connections, i.e.  $d\omega+\omega\omega=0$. 

Given this data one can consider an associated bundle whose fibre is the above space of ``functions'' in $Y,P$, where $o(n,2)$ acts according to the twisted representation~\eqref{twisted}, where at a given point $V_0$ is just $V$ at this point. Moreover we assume that the local frame is chosen in such a way that $V^A$ is constant. For instance, the associated covariant derivative of a section $\Phi$ is given explicitly by:%\maxim{Are signs consistent with the rest?} - Yes
\begin{equation}
 \pmb{\nabla}_\mu \Phi=\dl{x^\mu}\Phi+\omega^A_{\mu B}\left[P_A\dl{P_B}-(Y^B+V^B)\dl{Y^A}\right]\Phi\,.
\end{equation}
Here and in what follows we use the local identification of sections of this bundle with functions in  $x,P,Y$. In particular, for $Y$-indpendent sections the usual covarinat derivative is reproduced.  In the conformally flat case we are concerned with $\pmb{\nabla}$ is flat, i.e. $\pmb{\nabla}^2=0$. 

The standard choice of the local frame of $\cV(M)$ is such that:
\begin{equation}
\label{bulky_equation}
\begin{gathered}
\omega_{\mu\enspace B}^{\enspace A} = 
\begin{pmatrix}
  0 & -e_{\mu b} & 0 \\
  J_{\mu}^{\enspace a} & \omega_{\mu \enspace b}^{\enspace a} & e_{\mu}^{\enspace a} \\
  0 & -J_{\mu b} & 0 
 \end{pmatrix}
\,,
\qquad
\quad
V^A = \begin{pmatrix}
 V^-\\
 V^a\\
 V^+	
 \end{pmatrix}
 = \begin{pmatrix}
 0\\
 0\\
 1	
 \end{pmatrix}\,,
 \end{gathered}
 \end{equation}
where $e^A=\nabla V^A=\omega^A{}_B V^B$ and $J_{\mu}^{\enspace a} = e^{a\nu}J_{\mu\nu}$ with  $J_{\mu\nu}$ being the Schouten tensor of the metric $g_{\mu\nu}=e^A_\mu e^B_\nu \eta_{AB}$. In dimensions $n\geq 3$ the Schouten tensor is defined in terms of the Ricci tensor and the scalar curvature as $J_{\mu\nu} = \frac{1}{n-2}\left(R_{\mu\nu} - \frac{R}{2(n-1)}g_{\mu\nu} \right)$. We denote $J$ to be the trace of $J_{\mu\nu}$.

Although the construction is frame-independent we assume for simplicity  that the frame is chosen as above.
Note that it's also possible to allow for non-constant $V^A$ at the price of extra terms in the covariant derivative,
see \cite{Barnich:2006pc, Bekaert:2009fg} for more details.

With the above prerequisites we are ready to give a local version of the ambient definition of tractors. More precisely, consider the following system:
\begin{equation}
\label{parentsystem}
\begin{gathered}
\pmb{\nabla}_{\mu}\Phi = 0\,,\\
((Y+V)\cdot\frac{\partial}{\partial Y} - w)\Phi = 0\,,\\
\Phi \sim \Phi + (Y+V)^2 \chi\,,
\end{gathered}	
\end{equation}
where $\chi=\chi(x,P,Y)$ satisfies analogous system with $w$ replaced by $w-2$. The space of equivalence classes of sections determined by~\eqref{parentsystem} is precisely the space of tractors of weight $w$ that we keep denoting $\cE^\bullet[w]$.

The easiest way to see this is to observe that any $Y$-independent section $\Phi_0(x,P)$ admits a unique (up to an equivalence) lift to $\Phi(x,P,Y)$ satisfying~\eqref{parentsystem}. Indeed, taking into account the explicit form 
\eqref{bulky_equation} one finds that the 1-st and the 2-nd equations are first-order in $y^a,Y^+$ and hence solution exists and can be constructed recursively.~\footnote{More precisely, to give a rigorous argument it is useful to introducing Grassmann-odd ghost variables associated to all the constraints in~\eqref{parentsystem} (note that those associated to the components of the covariant derivative are precisely the basis differentials $dx^\mu$) and to employ the homological perturbation theory. In this way it is manifest that consistency conditions are fulfilled at each step.} The arbitrariness in the solution at each step is in adding a function in $Y^-$ but this arbitrariness is taken into account by the equivalence relation in the last line of~\eqref{parentsystem}. 

The GJMS operators can also be defined in terms of~\eqref{parentsystem} as
\begin{equation}
P^{2\ell} \Phi(x,P,Y)= \Box^{\ell} \, \Phi(x,P,Y)\,, \qquad 
\Phi\in \cE^{\bullet}[\ell-\frac{n}{2}]\,,\qquad 
\Box\coloneqq (\dl{Y}\cdot\dl{Y})\,.
\end{equation}
To recover the previous definition of GJMS operators in terms of $Y$-independent fields let $\Phi(x,P,Y)$
be a unique (up to equivalence) solution to~\eqref{parentsystem} with $w=\ell-\frac{n}{2}$ such that $\Phi(x,P,Y)|_{Y=0}=\Phi_0(x,P)$. Then
\begin{equation}
P^{2\ell} \Phi_0(x,P)= \left(\Box^{\ell} \, \Phi(x,P,Y)\right)\Big|_{Y=0}\,.
\end{equation}

The parent analog of the system \eqref{dualGJMS} describing a scalar of weight $\ell-\frac{n}{2}$ subject to GJMS equation of order $2\ell$ reads as
\begin{equation}
\label{dual_GJMS}
	\begin{gathered}
\pmb{\nabla}_{\mu}\Phi = 0\,,\\
((Y+V)\dl{Y} - \ell + \frac{n}{2})\Phi = 0\,,\\
\Box\Phi = 0\\
\Phi \sim \Phi + (Y+V)^{2\ell} \chi\,,
\end{gathered}	
\end{equation}
where $(Y+V)^{2k} \coloneqq ((Y+V)^2)^k$. In this system the GJMS equation arises at order $(Y^-)^{\ell-1}$ of $(\Box \Phi)\big|_{Y^+= Y^a = 0} = 0$.

There exists a partially gauge-fixed version of \eqref{parentsystem}. More precisely, for a given $\tilde\Phi$ satisfying \eqref{parentsystem} one can construct an equivalent representative $\Phi$ of the same equivalence class such that 
\begin{equation}
\label{pgf-tract}
	\Box \Phi= (Y+V)^{2(\ell-1)}\alpha
\end{equation}
for some $\alpha(x,P,Y)$. It is easy to see that the new representative is defined up to a restricted equivalence relation:
\begin{equation}
\label{new-equiv}
 \Phi\sim \Phi+(Y+V)^{2\ell}\chi\,.
\end{equation} 
More formally, \eqref{parentsystem} is equivalent to the partially gauge-fixed system consisting of the first two equations of  
\eqref{parentsystem} supplemented by \eqref{pgf-tract} and the new equivalence relation~\eqref{new-equiv}. In terms of 
\eqref{parentsystem},\eqref{pgf-tract} the GJMS equation can be written as $\alpha=(Y+V)^{2}\beta$ for some $\beta$.
This makes manifest that GJMS equation shows up at order $\ell-1$ in the expansion of $\Box\Phi$ in powers of $(Y+V)^2$.

In terms of tractors described through~\eqref{parentsystem}, Thomas D-operator can be defined as follows:
\begin{equation}
\label{TomasD}
	\mathcal{D}_A \Phi(x,P,Y)\coloneqq
	\left(2((Y+V)\cdot \frac{\partial}{\partial Y} + \frac{n}{2})\frac{\partial}{\partial Y^A} - (Y+V)_A\frac{}{}\Box\right) \Phi(x,P,Y)\,, 
\end{equation}
where $\Phi\in \cE^{\bullet}[w]$. It is easy to check that $P\cdot\mathcal{D}$ is a well-defined map $\cE^{\bullet}[w]\to \cE^{\bullet}[w-1]$. The explicit relation between $\cD_A$ and conventional Thomas-D operator  defined on tractors reads as
\begin{equation}
\label{ThomasD}
 D_A \Phi_0(x,P)=\left(\mathcal{D}_A \Phi(x,P,Y)\right)|_{Y=0}\,, 
\end{equation} 
where as usually $\Phi$ denotes a lift of $\Phi_0$ satisfying~\eqref{parentsystem}. In particular, this gives an alternative systematic way to derive  the explicit  expression for Thomas-D operator. The details of the derivation as well as the explicit expression for $D_A$ are given in Appendix~\bref{sec:components}. Note that a version of this derivation was in~\cite{Grigoriev:2011gp}.

The relation~\eqref{D->P1} between Thomas-D and GJMS operators take the form:
\begin{equation}
	\label{parentGJMS}
	\mathcal{D}_{A_1}... \mathcal{D}_{A_{\ell}}\Phi(x, P, Y) = (-1)^{\ell}(Y+V)_{A_1}...(Y+V)_{A_{\ell}}\Box^\ell\Phi(x, P, Y)\,.
\end{equation}

\subsection{Scale tractor and factorization of GJMS operators}

It is known that the GJMS operator factorises on a conformally-Einstein (in particular, a conformally-flat) background~\cite{Gover:2005mn}. This factorization is easy to arrive at explicitly by making use of additional
important ingredient, the so-called scale tractor. By definition, a scale tractor is a nowhere vanishing weight $0$ and rank $1$ tractor tensor $I^A$ which is parallel,
i.e. satisfying covariant-constancy condition:
\begin{equation}
\label{parallel}
 \nabla_\mu I^A=\partial_{\mu} I^A + \omega_{\mu\enspace B}^{\enspace A}I^{B}=0\,.
\end{equation}

In the conformally-Einstein case~\eqref{parallel} implies that there exists a scalar density $\sigma$ such that $I^A=\frac{1}{n}D^A\sigma$, where $D^A$ is a Thomas-D derivative determined by~\eqref{ThomasD}. In terms of components
\begin{equation}
	I^A = \begin{pmatrix}
		\sigma \\
		\bar\nabla^a \sigma\\
		-\frac{1}{n}(J + \bar\nabla^2)\sigma
	\end{pmatrix}
\end{equation}
where $\bar\nabla_a\coloneqq e_a^\mu \bar\nabla_\mu$. Here $\bar\nabla_\mu$ denotes Levi-Civita covariant derivative determined by the metric $g_{\mu\nu}=\nabla_\mu V^A \nabla_\nu V^B \eta_{AB}$.

If $I^A=\frac{1}{n}D^A\sigma$ and $\nabla_\mu I^A=0$ then $\sigma$ determines a constant curvature representative of the conformal equivalence class of the metric.  More precisely, $g^{c}_{\mu\nu}\coloneqq \sigma^{-2} g_{\mu\nu}$ is constant curvature \cite{LeBrun85}. In the case where $g_{\mu\nu}$ is constant curvature from the very beginning one may simply take $\sigma=1$ so that 
\begin{equation}
	I^A = \begin{pmatrix}
		1\\
		0\\
		-\frac{J}{n}
	\end{pmatrix}\,.
\end{equation}
Note that $V^AI_A =1$ with this choice. Although all the general constructions of this and the next section are valid for general parallel $I^A$ in all the explicit examples we always assume that the metric is constant curvature from the very beginning ($AdS$ for definiteness) and $\sigma=1$.

Because $I$ commutes with $\nabla_\mu$, it also commutes with Thomas D: $[I_A,D_B] = 0$. It is easy to see that $I^AP_A$ satisfies \eqref{parentsystem} with $w=0$ and hence $I_A$ also commutes with $\cD_B$.
It means that the formula for GJMS operator \eqref{parentGJMS}  can be written slightly differently~\cite{Gover:2005mn}:
\begin{equation}
\begin{gathered}
\label{useless}
	P^{2\ell} \Phi(x,P) =  (-1)^{\ell}I^{A_1}{D}_{A_1}...I^{A_{
	\ell}}{D}_{A_{\ell}} \Phi(x,P)\,,\qquad 
\end{gathered}
\end{equation}
Note that $\Phi(x,P)$ is to be understood as a tractor field of weight $w=\ell-\frac{n}{2}$.
Note also that $(I\cdot\mathcal{D})\Phi(x, P)$ has weight $w-1$.

Using~\eqref{useless} one can obtain  explicit formulas for GJMS operators. Indeed, for a tractor field  $\Phi_w(x,P)$ of generic weight $w$ one has 
\begin{equation}
\label{ID}
 (I\cdot \cD)\Phi_w(x,P)=-\lbrace \nabla^2 +\frac{2J}{n}(n+w-1)(w)\rbrace \Phi_w(x,P)\,.
\end{equation} 
Here and in what follows $\nabla^2=g^{\mu\nu}\nabla_\mu\nabla_\nu$ where by slight abuse of notations we denote by $\nabla_\mu$ the covariant derivative $\hat{\nabla}_\mu$ extended to tensors with values in symmetric tractors (identified with polynomials in $P_A$). For instance, for $A_\mu=A_\mu(x,P)$ one has $\nabla_\mu A_\nu=\hat{\nabla}_\mu A_\nu-\Gamma_{\mu\nu}^\rho A_\rho$, where $\Gamma_{\mu\nu}^\rho$ are coefficients of the Levi-Civita connection.
Combining this with~\eqref{useless} one gets:
\begin{equation}
\label{GJMS2}
\prod_{i=0}^{\ell-1} \lbrace \nabla^2 +\frac{2J}{n}(\ell+\frac{n}{2}-i-1)(\ell - \frac{n}{2}-i)\rbrace \Phi(x,P) = P^{2\ell}\Phi(x,P)\,.
\end{equation}
According to \eqref{ID} the order of terms entering \eqref{GJMS2} is the following: the term with $i=0$ acts first, then $i=1$, and so on. 

\section{Conformal higher spin fields}
\label{sec:CHS}
\
There exist conformally invariant equations for totally symmetric tensor fields proposed originally by Fradkin and Tseytlin~\cite{Fradkin:1985am} in 4 dimensions and extended to all even dimensions in~\cite{Segal:2002gd}. The equations are Lagrangian and possess gauge invariance. In Minskowski space the equations and gauge transformations have the following structure:
\begin{equation}
 (\d^2)^{\frac{n-4}{2}+s}\phi_{a_1\ldots a_s}+\ldots=0\,, \qquad \delta\phi_{a_1\ldots a_s}=\d_{(a_1}\epsilon_{a_2\ldots a_s)}+\eta_{(a_1a_s}\omega_{a_3\ldots a_s)}\,,
\end{equation} 
where $\ldots$ denote terms proportional to $\d^{a_1}\phi_{a_1\ldots a_s}$ and $\phi^{a_1}_{a_1\ldots a_s}$. The algebraic gauge symmetry with parameter $\omega$ can be employed to set $\phi^{a_1}_{a_1\ldots a_s}=0$. In what follows we always assume this gauge condition.

CHS fields can be seen as a linearization of the nonlinear CHS theory proposed in~\cite{Tseytlin:2002gz,Segal:2002gd} (see also~\cite{Bekaert:2010ky,Bonezzi:2017mwr}) about a Minkowski space vacuum. The respective action functional arises as an induced action for the scalar field in the higher-spin background.  It is remarkable, that the consistency of the scalar in higher-spin background naturally determines nonlinear gauge transformations for the background higher-spin fields, giving the gauge symmetries of the induced action~\cite{Segal:2002gd}.

In addition to CHS fields whose gauge transformations are of first order in derivatives there are conformal gauge fields whose gauge transformations are of order $t\leq s$ in derivatives, which we refer to as depth-$t$ CHS fields. The law-spin fields of this type were already in~\cite{Drew:1980yk,Barut:1982nj,Deser:1983mm} while higher spin ones were described much later~\cite{Erdmenger:1997wy,Vasiliev:2009ck,Bekaert:2013zya,Beccaria:2015vaa}. In what follows we use the term ``CHS fields'' for the entire family including Fradkin-Tseytlin fields as well as their depth-$t$ generalizations. Note that higher-depth CHS fields in $n$-dimensions are somewhat similar to so-called partially-massless fields~\cite{Deser:1983mm,Deser:2001pe} and in fact can be understood as boundary values of the partially-massless fields on $AdS_{n+1}$~\cite{Bekaert:2013zya}.  It is also worth mentioning that in addition to totally symmetric fields there exist mixed symmetry CHS fields~\cite{Vasiliev:2009ck,Chekmenev:2015kzf} which remain beyond the scope of the present work.

\subsection{Manifestly $o(n, 2)$-invariant description of CHS fields}

Our goal is to study the structure of CHS  equations and, in particular, the factorization of the CHS wave operators. As a starting point of our analysis we use the ambient space formulation of the CHS equations, which is available in the literature. More precisely, CHS equations of motion can be encoded in the following system~\cite{Bekaert:2012vt,Bekaert:2013zya} (see also~\cite{Chekmenev:2015kzf}):
\begin{equation}
\label{FT}
	\begin{gathered}
	  \pmb{\nabla} \Phi = 0 , \qquad ((Y+V)\cdot\frac{\partial}{\partial Y}+\frac{n}{2}-\ell)\Phi = 0, \qquad (Y+V)\cdot\frac{\partial}{\partial P}\Phi = 0\,,\\
	\frac{\partial}{\partial Y}\cdot\frac{\partial}{\partial P}\Phi = 0\,,\qquad \frac{\partial}{\partial P}\cdot\frac{\partial}{\partial P}\Phi = 0,\qquad
	P\cdot\frac{\partial}{\partial P}\Phi = s\Phi\,,\\
		\end{gathered}
\end{equation}
\begin{equation}
\label{Box}
 \Box\Phi = 0\,,
\end{equation}
\begin{equation}
\label{equiv-l}
 \Phi \sim \Phi+(Y+V)^{2\ell}\chi\,,
\end{equation}
which is formulated in the setting of Section~\bref{sec:tractors-parent} and where $\ell=\frac{n}{2}+s-t-1$. The equivalence relation is to be understood as follows: two configurations are equivalent if their difference can be represented as $(Y+V)^{2\ell}\chi$ for some $\chi$. Note that it's not difficult to extract explicitly the conditions $\chi$ ought to satisfy: these also have the form~\eqref{FT},\eqref{Box} but with $\ell$ replaced with $-\ell$.

In terms of the above representation the CHS gauge transformations can be written as follows:
\begin{equation}
\label{gt-PdY}
\begin{gathered}
\delta \Phi = (P\cdot\frac{\partial}{\partial Y})^t \epsilon\,,
\end{gathered}
\end{equation}
where $\epsilon=\epsilon(x,P,Y)$ is subject to the analogous system with $\ell$ replaced by $\ell+t$
and $s$ with $s-t$.
For $t=1$ system~\eqref{FT}-\eqref{gt-PdY} describes the usual CHS fields while for $t=2,\ldots,s$ their higher-depth generalizations. The system is manifestly $o(n, 2)$ invariant.

Strictly speaking system \eqref{FT}-\eqref{gt-PdY} is not equivalent to CHS equations of motion and gauge symmetries. More precisely, in addition to CHS equations
it also encodes conformal gauge conditions which, however, can be consistently removed in one or another way, giving an equivalent formulation of CHS fields (see~\cite{Bekaert:2012vt,Bekaert:2013zya,Chekmenev:2015kzf} for more details).

To see how exactly CHS equations are encoded in the above system let us consider~\eqref{FT} supplemented with 
\begin{equation}
\label{Box-lrel}
\Box \Phi=(Y+V)^{2(\ell-1)}\alpha
\end{equation} 
in place of~\eqref{Box}. By rephrasing the analysis of~\cite{Bekaert:2012vt,Bekaert:2013zya,Chekmenev:2015kzf} in the present terms one finds that any traceless $\phi(x,p)$ can still be lifted to $\Phi(x,P,Y)$ satisfying not only \eqref{FT} but also \eqref{Box-lrel}. Moreover, the condition that
%\maxim{How to see the conformal invariance of this equation?} \textcolor{blue}{something is written below}
\begin{equation}
\alpha =(Y+V)^2\beta
\end{equation}
for some $\beta$ encodes CHS equation  and conformal gauge conditions.  Here $\alpha$ is understood as a function of $\phi$ and its $x^\mu$-derivatives obtained by solving \eqref{FT},\eqref{Box-lrel}.

It turns out that just CHS equations can be written as 
\begin{equation}
\label{CHS}
 \alpha\big|_{Y=P^\pm=0}=0\,.
\end{equation}
The equation $\alpha\big|_{Y=0}=0$ is conformal by construction.
To check the conformal invariance of \eqref{CHS} one observes that the equation sitting at $P^{\pm}=0$ is of order $2\ell$ while the equations in different components are of higher order and hence their conformal transformations can't compensate the transformations of \eqref{CHS}  so that \eqref{CHS} should be conformally invariant~\cite{Chekmenev:2015kzf}. Their gauge invariance can also be shown on general grounds following~\cite{Bekaert:2009fg,Bekaert:2013zya,Chekmenev:2015kzf}. In any case in this work we give an independent proof of the gauge invariance.

It is also worth mentioning that if one drops the equivalence relation in~\eqref{FT}, takes $V$ such that $V^2=-1$, and takes as $\Phi$ a field on  $AdS_{n+1}$ rather than $n$-dimensional conformal space, the above system   
is precisely the one from~\cite{Alkalaev:2011zv} see also~\cite{Barnich:2006pc,Alkalaev:2009vm}, which describes partially-massless fields on $AdS_{n+1}$. In this form it is manifest that depth-$t$ FT fields in $n$-dimensions are boundary values of the partially-massless fields on $AdS_{n+1}$.

There exist a ``dual'' system that also describes CHS fields but where  
the harmonicity condition is replaced by $(\frac{\partial}{\partial Y}\cdot\frac{\partial}{\partial Y})^{\ell}\Phi(x, P, Y) = 0$:
\begin{equation}
\label{FT1}
	\begin{gathered}
	\pmb{\nabla} \Phi = 0 , \qquad ((Y+V)\cdot\frac{\partial}{\partial Y}+\frac{n}{2}-\ell)\Phi = 0, \qquad (Y+V)\cdot\frac{\partial}{\partial P}\Phi = 0\,,\\
	\frac{\partial}{\partial Y}\cdot\frac{\partial}{\partial P}\Phi = 0\,,\qquad 
	\frac{\partial}{\partial P}\cdot\frac{\partial}{\partial P}\Phi = 0\,,\qquad 
	P\cdot\frac{\partial}{\partial P}\Phi = s\Phi\,,
		\end{gathered}
\end{equation}
\begin{equation}
\label{Box-l}
\Box^{\ell}\Phi = 0 \,,                                                                                    
\end{equation} 
\begin{equation}
\label{equivalence-2}
\Phi \sim \Phi+(Y+V)^2\alpha\,.
\end{equation} 
Note that \eqref{FT1} is identical to \eqref{FT}.

In what follow it is useful to introduce a natural map $L^{-1}$ that sends elements of $\cE^\bullet[w]$ (in particular solutions to~\eqref{FT1}) to tensor fields on $M$. In terms of generating function $\Phi(x,P,Y)$ it is given by
\begin{equation}
L^{-1}\Phi=\Phi\big|_{Y=P^{\pm}=0}\,.
\end{equation} 
We have the following:
\begin{prop}
\label{prop:lift}
For all $\ell>0$ or all non-integer $\ell$ any $\phi(x, p)$ satisfying $\dl{p}\cdot \dl{p}\phi=0$ and $p\cdot\dl{p}\phi=s\phi$ can be lifted to $\Phi(x,P,Y)$ satisfying~\eqref{FT1} and such that $\phi=L^{-1}\Phi$.
% \begin{equation}
% \phi=L^{-1}\Phi%\left(\Phi(x,P,Y)\right)\big|_{Y=P^{\pm}=0}\,.%\,, \quad L^{-1}F(x,P,Y)\coloneqq F(x,P,Y) \Big|_{Y=P^{\pm}=0}
% \end{equation}
The lift is unique if one takes into account the equivalence relation~\eqref{equivalence-2}. 
\end{prop}
The statement can be inferred from the analysis of~\cite{Bekaert:2012vt,Bekaert:2013zya,Chekmenev:2015kzf}. Some details of the proof are also given in Appendix~\bref{sec:lift-obst}. 
\begin{prop}
Let $\Phi(x,P,Y)$ be a lift of $\phi(x,p)$ as described in Proposition~\bref{prop:lift} with $\ell=\frac{n}{2}+s-t-1$, then the operator defined by
\begin{equation}
\label{CHS-operator}
A_{s,t}\phi=L^{-1}\left(\Box^{\ell}\Phi\right)
\end{equation}
is well-defined on equivalence classes~\eqref{equivalence-2} and coincides with CHS wave operator.
\end{prop}
\begin{proof}
Operator $A_{s,t}$  is clearly well-defined on equivalence classes $\Phi \sim \Phi+(Y+V)^2\alpha$ and hence determines an operator of derivative order $2\ell$ on totally symmetric traceless tensor fields.

To explicitly relate $A_{s,t}$ to CHS operator we first note that for any $\Phi_0$ satisfying~\eqref{FT1} one can find an equivalent element $\Phi=\Phi_0+(Y+V)^2\ldots$ such that 
\begin{equation}
\label{Box-ell-rel}
 \Box \Phi=(Y+V)^{2(\ell-1)}\alpha\,.
\end{equation} 
Indeed, this is achived by taking
\begin{equation}
\begin{gathered}
\Phi(x, P, Y) = \sum_{k=0}^{l-1}(Y+V)^{2k}\Phi_{k}\\
\Phi_{k} = -\frac{1}{4k(\ell - k)}\Box\Phi_{k-1}\,.
\end{gathered}
\end{equation}
Note that the residual equivalence relation is precisely~\eqref{equiv-l}. In this way we have found that the system~\eqref{FT},\eqref{equiv-l},\eqref{Box-lrel} results from~\eqref{FT1},\eqref{equivalence-2} by partially taking into account the equivalence relation~\eqref{equivalence-2} and hence these systems are equivalent.

Finally, applying $\Box^{\ell-1}$ to both sides of~\eqref{Box-ell-rel} and setting to zero $Y,P^+,P^-$
one finds 
\begin{equation}
 L^{-1}(\Box^\ell \Phi)=r\,L^{-1}(\alpha)
\end{equation} 
where $r$ is a non-vanishing coefficient
%\maxim{check!} \textcolor{blue}{True} 
and hence equation $L^{-1}(\Box^\ell \Phi)=0$
is equivalent to CHS equations~\eqref{CHS}.
\end{proof}

Let us comment on the relation between the above description of CHS fields and tractors. 
In contrast to tractor fields, which can be seen as certain tensor fields on the $n+2$-dimensional ambient space restricted to $n$-dimensional submanifold, CHS fields (at the off-shell level, i.e. before imposing CHS equations of motion)  are tensor fields (more precisely, tensor densities) in $n$-dimensions on which the action of gauge transformations and conformal transformations is defined. Equations~\eqref{FT1},\eqref{equivalence-2} can be seen as a mean to embedd off-shell CHS fields as a subspace of tractor fields in such a way that the GJMS operator produces the CHS equations of motion through~\eqref{CHS-operator}.

\subsection{Modified system and factorization of CHS operators} 

It turns out that it is useful to employ a certain modification of the system~\eqref{FT1},\eqref{equivalence-2}. In particular, the gauge invariance of the CHS equations is conveniently analysed in the modified formulation. Consider the following system:
\begin{equation}
\label{FT2}
	\begin{gathered}
	\pmb{\nabla} \Phi = 0 , \qquad ((Y+V)\cdot\frac{\partial}{\partial Y} -w)\Phi = 0, \qquad  (Y+V)\cdot\frac{\partial}{\partial P}\Phi = 0\\
	 \mathcal{D}\cdot\frac{\partial}{\partial P}\Phi = 0\,, \qquad 
	\frac{\partial}{\partial P}\cdot\frac{\partial}{\partial P}\Phi = 0,\qquad 
	P\cdot\frac{\partial}{\partial P}\Phi = s\Phi\,,
\end{gathered}
\end{equation}
\begin{equation}
\label{equiv-22}
 	\Phi \sim \Phi + (Y+V)^2\chi\,,
\end{equation} 
where $\chi(x, P, Y)$  satisfies \eqref{FT2} with $w$ replaced by $w-2$.
We denote by $S[s, w]$ the space of equivalence classes determined by this system.
The gauge transformations can be now defined as $\delta\Phi = (P\cdot\mathcal{D})^{t}\epsilon$, where $\epsilon$ also satisfies \eqref{FT2},\eqref{equiv-22} with $w,s$ replaced with $w+t,s-t$ 
so that $(P\cdot\mathcal{D})^{t}$ determines a well-defined map $S[s-t, w+t]\to S[s, w], \quad w = s-t-1$ 

It is easy to check that for $w\neq -\frac{n}{2}$ equations \eqref{FT2} are equivalent to \eqref{FT1}. Indeed,
assuming all the other constraints in \eqref{FT2} but $\cD\cdot\dl{P}$ satisfied one finds:
\begin{equation}
\cD\cdot\dl{P}=(2w+n)\dl{P}\cdot\dl{Y}\,.
\end{equation} 
In particular, for $w=\ell-\frac{n}{2}$ if $\Phi(x,P,Y)$ satisfying \eqref{FT2} denotes a lift of $\phi(x,p)$ ($\dl{p}\cdot \dl{p})\phi=0$ then \eqref{CHS-operator} defines CHS wave operator.

For the special value $w=-\frac{n}{2}$ any $\phi(x,p)$ satisfying $\dl{p}\cdot\dl{p}\phi=0$ can be lifted to $\Phi(x,P,Y)$ satisfying \eqref{FT2}. However,
in contrast to \eqref{FT1} the lift is not unique even if one takes into account the equivalence relation~\eqref{equiv-22}. The uniqueness can be restored by introducing  the following additional equivalence relation:
\begin{equation}
\label{extra-equiv}
 \Phi \sim \Phi+ ((Y+V)\cdot P)\beta\,.
\end{equation} 
Assuming that definiton of $S[s,w]$ in the case of $w=-\frac{n}{2}$ also involves~\eqref{extra-equiv} we conclude that for $w \geq -\frac{n}{2}$, the space $S[s,w]$ is one-to-one with that of totally symmetric traceless tensor fields on $M$.

The apparent disadvantage of defining CHS operator through~\eqref{CHS-operator} is that it requires extracting particular components of $\Box^\ell\Phi$. 
This can be cured as follows: pick a particular metric in the conformal class and consider the following operator defined on $S[s, s-k-1]$:
\begin{equation}
	B_k \coloneqq I \cdot\mathcal{D} - \frac{1}{k}(P\cdot\mathcal{D})(I\cdot \dl{P})
\end{equation}
Indeed, it is well defined on the equivalence classes~\eqref{equiv-22} for these values of parameters. If $\Phi(x, P, Y) \in S[s, s-k-1]$, then $(B_k\Phi)(x, P, Y) \in S[s, s-k-2]$, so that $B_k$ determines a well defined map $S[s, s-k-1]\to S[s, s-k-2]$. Note that in contrast to the operators employed above $B_k$ is not $o(n,2)$-invariant because it contains
the scale tractor that breaks $o(n,2)$-symmetry. With our choice of $I^A$ the residual symmetry is just $(A)dS$-isometries. Note that for $k=1$ this operator was employed in~\cite{Grigoriev:2011gp}, while for $k=1$ and $s=1,2$ it was in~\cite{Gover:2008sw}.

One can also define powers of $B_k$ as follows: $B^{\ell}_{k}\coloneqq B_{k+l-1}\circ ... \circ B_{k}$ which act on $\Phi\in S[s,s-k-1]$ according to  $B^{\ell}_{k}: \Phi \mapsto (B_k^{\ell}\Phi) \in S[s,s-k-1-{\ell}]$. We have the following:
\begin{prop}
\label{proposition}
Let $\Phi(x, P, Y)\in S[s, w]$ with $\quad w = \ell - \frac{n}{2}, \ell = \frac{n-2}{2}+s-t$ be a lift of $\phi(x,p)$, i.e. $L^{-1}\Phi=\phi$. Then equation 
$B^{\ell}_{t}\Phi = 0$ is equivalent to $A_{s,t}\phi=0$ and hence is a CHS equation formulated in terms of $S[s, w]$.
\end{prop}
It is clear that $B_t^{\ell}$ is well defined on $S[s, w]$ in this case. The proof that it indeed determines CHS wave operator is relegated to Appendix~\bref{sec:Blt-proof}. 

\def\n2{\frac{n}{2}}
It follows from the identification of $S[s,s-k-1]$ with totally symmetric tensor densities that $B_k$ defines an
operator on tensor densities. More precisely, if $L_{s,w}$ denote a map that sends $\phi(x,p)$ to $\Phi(x,P,Y)$ satisfying \eqref{FT2} and $L^{-1}$ the map defined by $L^{-1}\Phi=\Phi\big|_{Y^A,P^\pm=0}$ then for $\phi=\phi(x,p)$ of rank $s$ and weight $s-k-1$
\begin{equation}
\bar B_k\phi=L^{-1}B_k L_{s,s-k-1}\phi
\end{equation} 
is a second order differential operator on tensor densities. Representing $A_{s,t}$ as
\begin{multline}
A_{s,t}\phi=(-1)^{\ell}L^{-1}B_k^\ell L_{s,s-t-1}
=\\=
(-1)^{\ell}(L^{-1}B_{t+\ell-1} L_{s,s-t-\ell})( L^{-1}\ldots    L_{s,s-t-2})(L^{-1} B_t L_{s,s-t-1})
\end{multline} 
one finds
\begin{equation}
\label{A-exp}
 A_{s,t}\phi=(-1)^{\ell}\bar B_{t+\ell-1}\ldots \bar B_{t} \phi\,.
\end{equation} 
In other words we have arrived at the manifestly factorized form of the CHS wave operator. Note that although all the above arguments apply to generic conformally-flat background metric $g_{\mu\nu}$ operators $\bar B_k$ in general depend on scale $\sigma$ so that only on constant curvature spaces where one can take $\sigma=1$ this gives a genuine factorization of CHS wave operator into natural second-order operators.

\subsection{Explicit form of the factors}

Now we are ready to give an explicit component expressions for the CHS operators and the operators $\bar B_k$ in terms of tensor densities on $M$. Leaving the detailed computations for the Appendix~\bref{sec:lift-obst} we get 
 \begin{multline}
  	\bar B_k \phi(x,p)=L^{-1}B_k L_{s,s-k-1}\phi(x, p) = \\ =-
 	\lbrace \bar\nabla^2 +\frac{2J}{n}(-s + (n+s-k-2)(s-k-1) ) -
 	\frac{n+2s-4}{k(n+2s -k-3)}(p\cdot\bar\nabla)(\frac{\partial}{\partial p}\cdot\bar\nabla)
 	+\\+
 	\frac{1}{k(n+2s-k-3)}p^2(\frac{\partial}{\partial p}\cdot\bar\nabla)^2 \rbrace \phi(x,p)\,.
\end{multline}

 It follows from the structure of the mass-like term in the above operator that it coincides with the one of partially-massless field of spin $s$ and depth $k-1$. More precisely, $B_k$ explicitly coincides with the partially massless operator provided both are written in the gauge where $(\dl{p}\cdot\bar\nabla)  \phi(x,p)=0$.

Now the expression~\eqref{A-exp} for the CHS wave operator takes the form:
\begin{equation}
\label{factor_example}
\begin{gathered}
	A_{s,t}\phi(x,p)=\prod_{i=1}^{\frac{n-4}{2}+s -t +1}\lbrace \bar\nabla^2 +\frac{2J}{n}(-s + (n+s-t-i-1)(s-t-i) ) -\\- \frac{n+2s-4}{(t+i-1)(n+2s-t-i-2)}(p\cdot\bar\nabla)(\frac{\partial}{\partial p}\cdot\bar\nabla)+\\+\frac{1}{(t+i-1)(n+2s-t-i-2)}p^2(\frac{\partial}{\partial p}\cdot\bar\nabla)^2 \rbrace \phi(x,p)\,,
\end{gathered}	
\end{equation}
where the operator with $i=0$ acts first, then the operator with $i=1$ and etc. In the special case $w=s-2$ we get the result obtained by Nutma and Taronna \cite{Nutma:2014pua}:
\begin{equation}
\label{Nutma_Taronna}
\begin{gathered}
	A_{s,1}\phi(x, p)=\prod_{i=1}^{i = \frac{n-4}{2} +s}\lbrace \bar\nabla^2 +\frac{2J}{n}(-s + (n+s-i-2)(s-i-1) ) -\\- \frac{n+2s-4}{i(n+2s-i-3)}(p\cdot\bar\nabla)(\frac{\partial}{\partial p}\cdot\bar\nabla)+\frac{1}{i(n+2s-i-3)}p^2(\frac{\partial}{\partial p}\cdot\bar\nabla)^2 \rbrace \phi(x, p)\,.
\end{gathered}	
\end{equation}
Note that although the formulas coincide in our derivation we assumed that $\phi$ is traceless.

\subsection{Gauge invariance}

Now we are ready to analyse gauge invariance of the CHS equations using its factorized representation in terms of $B_k$. 
\begin{prop} 
For any $\epsilon \in S[s-t,s-1]$
\begin{equation}
 B^\ell_t (P\cdot \cD)^t\epsilon=0\,.
\end{equation} 
\end{prop}
As $S[s-t,s-1]$ is one-to-one with traceless tensor densities on $M$, $(P\cdot \cD)^t$ determines a gauge symmetry of the CHS equations.

\begin{proof}
Observe that 
\begin{equation}
\label{gauge_transf}
\begin{gathered}
	B_t(P\cdot \mathcal{D})^t\epsilon(x, P, Y) = (I\cdot\mathcal{D} - \frac{1}{t}(P\cdot\mathcal{D}) (I\cdot\dl{P})(P\cdot\mathcal{D})^t\epsilon(x, P, Y) =\\
	=-\frac{1}{t}(P\cdot\mathcal{D})^{t+1}(I\cdot\dl{P})\epsilon(x, P, Y)
	\end{gathered}\,.
\end{equation}
Applying $B_{t+1}$ we get $(P\cdot\mathcal{D})^{t+2}(I\cdot\frac{\partial}{\partial P})^2\epsilon(x, P, Y)$ and so on. Because $\epsilon$ is of rank $s-t$, this procedure gives zero after $s-t+1$ iterations. CHS operator $B^{\ell}_t,\,\, \ell= \frac{n-2}{2} +s -t$ contains at least $s-t+1$ factors (for $n\geq 4$) and hence $(P\cdot\mathcal{D})^t\epsilon(x, P, Y)$ is in the kernel of $B^\ell_t$.

To make sure that the gauge transformation $\delta \Phi=(P\cdot\cD)^t\epsilon$ indeed coincides with the standard gauge transformation for CHS fields, one can  check that $L^{-1}(P\cdot\cD)^t\Phi$ is traceless by construction and the leading term is proportional to $(p\cdot\bar\nabla)^t$ as it should be for CHS fields of this type.
\end{proof}

The above technique can be also used to study gauge invariance of $B_k$. Suppose that we subject the gauge parameter $\epsilon\in S[s-k,s-1]$ to the extra condition $I\cdot\dl{P}\epsilon = 0$ which encodes that $L^{-1}\epsilon$
satisfies $\dl p \cdot \bar\nabla(L^{-1}\epsilon)=0$. Then $B_t(P\cdot\mathcal{D})^t\epsilon = (P\cdot\mathcal{D})^{t+1}I\cdot\frac{\partial}{\partial P}\epsilon = 0$.  This gives an additional argument that for $\Phi\in S[s,s-k-1]$ equation $B_k\Phi=0$ is a partially gauge-fixed version of the equations of motion of the partially-massless field of spin $s$ and depth $k$. Examples can be found in Appendix \ref{sec:gauge_pm}.
\subsection{CHS equations in terms of tractors}

Although we have described CHS fields and found the factorized form of the CHS equations by employing the parent formalism, it turns out that the resulting formulas can be written in terms of usual tractor fields. To see this let us find the conditions satisfied by
$\Phi_0(x,P)=(L_{s,w}\phi(x,p))\big|_{Y=0}$. These can be easily obtained by setting $Y^A=0$ in  \eqref{FT2}, giving
\begin{equation}
\label{embed_tractors}
	\begin{gathered}
		V\cdot\dl P \Phi_0 = 0 
		\qquad P\cdot\dl{P}\Phi_0=s\Phi_0\,, \qquad \dl{P}\cdot \dl{P}\Phi_0=0\,,
		\\ 
		\dl{P}\cdot{D} \,\,\Phi_0 = 0\,,
	\end{gathered} 
\end{equation}
where $D_A$ is the usual Thomas-D derivative whose definition and explicit expression are given in respectively~\eqref{ThomasD} and \eqref{Thomas-D-comp}. 

It turns out, that for relevant values of $w$ one can avoid constructing $\Phi=L_{s,w}\phi$ and obtain $\Phi_0$ directly by solving \eqref{embed_tractors} with boundary condition $\Phi_0\big|_{P^\pm=0}=\phi(x,p)$. More precisely,  
with our choice of $V$ the first equation implies $\dl{P^{+}} \Phi_0(x, P)=0$ and hence the last one uniquely fixes the $P^-$ dependence.  Indeed, for $w \neq -\frac{n}{2}$ the last equation is equivalent to $\left( \dl{P^-}(w-1+ n+s - P^-\dl {P^-}) + \dl{p} \cdot \bar\nabla) \right) \Phi(x,P) = 0$ and hence always admits a unique solution if $w$ corresponds to a CHS field. 

It follows \eqref{embed_tractors} determines a particular embedding of totally symmetric traceless tensor densities into traceless symmetric tractors.  Identifying off-shell CHS fields with  the weight $w=s-t-1$ tractor fields satisfying~\eqref{embed_tractors} the CHS equations of motion take the following form:
\begin{equation}
	\mathbb{B}_t^{\ell}\Phi(x, P) = 0,\qquad l = \frac{n-2}{2} +s -t\,,
\end{equation}
where $\mathbb{B}_t = I\cdot D - \frac{1}{t}(P\cdot D)(I\cdot {\dl P})$ is just a tractor version of $B_t$, i.e. where $\cD$ is replaced with the conventional Thomas-D operator, and $\mathbb{B}_t^{\ell}\coloneqq \mathbb{B}_{t+\ell -1}\circ\mathbb{B}_{t+\ell -2}...\circ\mathbb{B}_{t}$. The operator $\mathbb{B}_t$ is well-defined on \eqref{embed_tractors} for $w=s-t-1$, i.e $\mathbb{B}_t\Phi_0$ satisfies \eqref{embed_tractors} for $w=s-t-2$. The gauge transformations are given by
\begin{equation}
\delta \Phi_0 = (P\cdot D)^t\epsilon 
\end{equation} 
where $\epsilon$ is a weight-$s-1$ and rank-$s-t$ tractor field satisfying~\eqref{embed_tractors} 
with $s$ replaced by $s-t$.

Let us consider as a simple example Maxwell field in 4 dimensions, i.e. $n=4,s=t=1$. 
Solving~\eqref{embed_tractors} with the initial condition $\Phi_0\big|_{P^\pm=0}=\phi^a p_a$ and analogous equations for the gauge parameter $\epsilon$ (these are satisfied trivially) gives:
\begin{equation}
\Phi_0^A = 
	\begin{pmatrix}
		0\\
		\phi^a\\
		-\frac{1}{2}\bar\nabla_a\phi^a
	\end{pmatrix}, \qquad D^A\epsilon = \begin{pmatrix}
		0\\
		2\bar\nabla^a\epsilon\\
		-\bar\nabla^2\epsilon
	\end{pmatrix}\,.
\end{equation}
Restricting for simplicity to the flat case and computing $\mathbb{B}_1\Phi_0$ gives
\begin{equation}
\label{Bmaxwell}
(\mathbb{B}_1\Phi_0)^A = 
	\begin{pmatrix}
		0\\
		\partial^2 \phi^a  - \partial^a\partial_b\phi^b\\
		0
	\end{pmatrix} 
\end{equation}
so that in accord with our general statements we indeed get just Maxwell equations. 

If instead of $\mathbb{B}_1\Phi_0$ we consider $\nabla^2 \Phi_0$ (which is precisely $(\Box \Phi(x,P,Y))|_{Y=0}$) we arrive at~\cite{Eastwood:1985eh}
\begin{equation}
 \nabla^2 \Phi = 
	\begin{pmatrix}
		0\\
		\partial^2 \phi^a  - \partial^a\partial_b\phi^b\\
		-\frac{1}{2}\partial^2\partial_a\phi^a
	\end{pmatrix} 
	\,.
\end{equation}
The second slot still contains Maxwell equations themselves, while the last one is the conformal gauge \cite{Eastwood:1985eh} also known as Eastwood-Singer gauge. Of course, this is the same gauge as encoded in the system~\eqref{FT},\eqref{Box},\eqref{equiv-l} for $s=1,t=1$ on top of the Maxwell equations. For general CHS fields one gets higher-spin analogs of this gauge.

\section*{Acknowledgements}
We are grateful to K.~Alkalaev, A.~Chekmenev, R.~Metsaev for useful discussions.
M.G. also wishes to thank N.~Boulanger, X.~Bekaert, and especially  A.~Waldron.  
This work was supported by the Russian Science Foundation grant 18-72-10123.

\appendix
\section{Component expressions}
\label{sec:components}
Here we compute an explicit expression for $\bar B_{s-w-1}=L^{-1}B_{s-w-1}L_{s, w}\phi$. To this end we first compute $\mathcal{D}_A\Phi(x, P, Y)\Big|_{Y=0}$ where $\Phi(x,P,Y)$ is a solution to~\eqref{parentsystem} with the initial condition $\Phi|_{Y=0}=\Phi_0(x,P)$. By using the freedom 
described by the equivalence relation in~\eqref{parentsystem}  $\Phi$ can be assumed $Y^-$-independent. From the first equation in \eqref{parentsystem} one may obtain $Y^a$-derivatives:
\begin{equation}
\label{A1}
	\begin{aligned}
		\left(\dl {Y^a}\Phi\right)\Big|_{Y=0} 
		&= e^\mu_a\nabla_\mu\Phi_0\qquad \\ 
		(\Box\Phi)\big|_{Y=0}=\left( \dl {Y^a} \dl {Y_a}\Phi\right)\Big|_{Y=0} 
		&= (g^{\mu\nu}\nabla_\mu\nabla_\nu + wJ)\Phi_0\,.
	\end{aligned}
\end{equation}
where $e^a_\mu e^b_\mu=\delta_\mu^\nu$. Recall that $\nabla_\mu$ is the covariant derivative $\hat{\nabla}_{\mu}$ extended to tensors with values in tractors and components of $\omega^A_{\mu B}$ introduced in~\eqref{bulky_equation} are given by:
\begin{equation}
\label{A3}
	\omega^a_{\mu +} = e^{a}_{\mu}, \qquad \omega^a_{\mu -} =  J^{a}_{\mu}\,, \qquad 
	\omega^-_{\mu b} = -e_{\mu b}, \qquad \omega^{+}_{\mu b} = -J_{\mu a}\,.
\end{equation}

The second equation determines $Y^+$-derivatives:
\begin{equation}
\label{A2}
\begin{gathered}
	\dl {Y^+} \Phi\Big|_{Y=0} = w\Phi_0(x, P)\,.
	\end{gathered}
\end{equation}

Using \eqref{A1}, \eqref{A2}, \eqref{A3} one finds (Here we retained derivatives in $Y^-$. One may observe that they vanish in the  following expression):
\begin{multline}
\label{Thomas-D-comp}
	\left(\cD_{A}\Phi\right)\Big |_{Y=0} 
	= \left[(n+2w-2)\dl {Y^A} - V_A(\nabla^2 + Jw+ (n+2w-2)\dl{Y^-})\right] \Phi\Big |_{Y=0} 
	=\\= \begin{pmatrix}
 (n+2w-2)w\\
 (n+2w-2)\nabla_a\\
 -(\nabla^2 + Jw)	
 \end{pmatrix}\Phi_0\,.
\end{multline}
The last expression is precisely the component form of Thomas-D derivative of $\Phi_0$ so that indeed~\eqref{ThomasD} reproduces Thomas-D derivative.

Let now $\Phi(x,P,Y)=L_{s, w}\phi$ for some $\phi=\phi(x,p)$ satisfying $p\cdot \dl{p}\phi=s\phi$ and $\dl{p}\cdot \dl{p}\phi=0$. Of course, $\Phi$ still satisfies \eqref{parentsystem} as \eqref{parentsystem} is just a part of \eqref{FT2}. Setting $P^\pm=0$ in~\eqref{Thomas-D-comp} 
and using extra constraints present in~\eqref{FT2} one gets
\begin{equation}
\label{A5}
 \begin{gathered}
 \left({\mathcal{D}_A}\Phi\right)\Big|_{Y=P^{\pm}= 0} =  
 \begin{pmatrix}
 (n+2w-2)w\\
 (n+2w-2)(e^{\mu}_{\enspace a}\bar\nabla_{\mu} +p_a \frac{\partial}{\partial P^-})\\
 -(\Box^{\dag}+ wJ - 2\frac{J}{n}s)
 \end{pmatrix}\Phi(x,P,Y=0)\Big |_{P^{\pm} = 0}\,,
 \end{gathered}
 \end{equation}
where $\Box^{\dag} = \bar\nabla^2 + 2(p\cdot\bar\nabla)\frac{\partial}{\partial P^{-}} + p^2 (\frac{\partial}{\partial P^{-}})^2$. Here and below we again abuse notations by identifying expansion coefficients in $p_a$ as tensor field on which $\bar\nabla$ acts as a Levi-Civita covariant derivative. Note that $\dl {P^{+}}\Phi(x,P,Y=0) = 0$

To compute $(P\cdot\mathcal{D})(I\cdot \dl P)\Phi$ note that $I\cdot \dl P: S[s, w] \mapsto S[s, w-1]$ is a well-defined map, and $I\cdot \dl P\Phi|_{Y=0} = \dl {P^-}\Phi(x, P,Y=0)$. Using this and \eqref{A5} one finds
\begin{multline}
	L^{-1}(P\cdot\mathcal{D})(I\cdot \dl P)\Phi= \\=
	(n+2w-2)\left((p\cdot\bar\nabla\dl{P^-} +p^2\frac{\partial^2}{(\partial P^-)^2} )\Phi(x, P,Y=0)\right)\Big|_{P^{\pm} = 0}\,.
\end{multline}
For $w \neq -\half n$ from 1st, 2nd, 3rd, and 4th equations in \eqref{FT2} one finds $P^-$ dependence:
\begin{gather}
\Phi(x, p, P^-) = \sum_{k=0}^{s}\frac{(P^-)^k}{k!}\Phi_k(x, p)\\
\lbrack \frac{\partial}{\partial P^{-}}(n+s+w-P^{-}\frac{\partial}{\partial P^{-}} - 1) + \frac{\partial}{\partial p}\cdot\bar\nabla \rbrack \Phi(x, p, P^{-}) = 0	
\label{P-plus}\\
\Phi_{1} = - \frac{1}{n+s+w-2}(\frac{\partial}{\partial p}\cdot\bar\nabla)\phi(x,p)\\
\Phi_{2} = \frac{1}{(n+s+w-3)(n+s+w-2)}(\frac{\partial}{\partial p}\cdot\bar\nabla)^2\phi(x,p)
\end{gather}
Finally, one finds
 \begin{multline}
 \label{A8}
 	L^{-1}B_{s-w-1}L_{s, w}\phi(x, p) = -\lbrace \bar\nabla^2 +\frac{2J}{n}(-s + (n+w-1)w ) -\\- \frac{n+2s-4}{(s-1-w)(n+s+w-2)}(p\cdot\bar\nabla)(\dl p\cdot\bar\nabla)
 	+\\+
 	\frac{1}{(n+s+w-2)(s-1-w)}p^2(\dl p\cdot\bar\nabla)^2 \rbrace \phi(x, p)\,.
 \end{multline}
Now one can easily obtain \eqref{factor_example} using \eqref{A-exp} and \eqref{A8} for weight $w =s-t-1$.

\section{Lifts and obstructions}
\label{sec:lift-obst}
Here we verify that the system \eqref{FT} 
(and hence  \eqref{FT2} unless  $w\neq -\frac{n}{2}$) is off-shell, i.e. it does not impose any equations on $\phi(x,p)=\Phi(x,P,Y)\big|_{Y=P^\pm=0}$ besides $p\cdot \dl{p}\phi=s\phi$ and $\dl{p}\cdot \dl{p}\phi=0$. To this end for a given $\phi(x,p)$ we construct a particular lift $\Phi(x, P, Y)$ satisfying \eqref{FT}. More precisely, we choose $\Phi(x, p, P^-, Y^-, P^+ = Y^a= Y^+ = 0)$ to be $Y^-$-independent and observe that $\Phi_0=\Phi|_{Y=0}$ satisfies \eqref{P-plus} as a consequence of the 4th equation in~\eqref{FT}. It is easy to check that  unless $w = t+1-s-n,\quad t =1, 2,..., s$
\eqref{P-plus} has a solution. Then by expansion in powers of $Y^a,Y^+,P^+$ the 1st, 2nd, and 3rd equations can be solved order by order, giving $\Phi(x,P,Y)$. It is then a matter of direct check that such $\Phi(x,P,Y)$ satisfies the 4th equation.

% For simplicity we choose a flat Minkowski metric on $M$. The lift can be obtained using the following recurrent conditions:
% \begin{equation}
% \label{components}
% 	\begin{gathered}
% 	\lbrace	\dl {P^+}(w-1+ n+s - P^+\dl {P^+}) +\dl p\dl x\rbrace \Phi(x, p, P^+) = 0\\	
% 		\Phi(x, p, P^+, P^-, Y^+) = \Phi(x, p, P^+) + \sum_{j=1}\frac{1}{j!}(-P^{-}Y^+)^{j} (\dl {P^+})^{j}\Phi(x, p, P^+)\\
% 		\Phi(x, P^+, P^-, Y^+, Y^-) = \sum_{k, j}(Y^-)^k(P^-Y^+)^{j}\Phi_{k, j}, \qquad \Phi_{(k+1),j} = \frac{1}{k+1}(w-j-k)\Phi_{k, j}\\
% 		\dl {Y^a}\Phi(x, P, Y) = \frac{1}{Y^- +1}(\dl {x^a} + p^a\dl{P^+} - P^-\dl{p^a} + Y^a\dl{Y^+})\Phi(x, P, Y)
% 	\end{gathered}
% \end{equation}
% Note that the lift is not unuiqe due to the equivalence relation $\Phi \sim \Phi +(Y+V)^2\chi$ the lift is not unique.
 
If $w = t+1-s-n,\quad t =1, 2,..., s$ the system \eqref{FT} is not off-shell.  Equation \eqref{P-plus} clearly implies the following condition on $\phi(x,p)$:
\begin{equation}
\label{new_constraint}
(\frac{\partial}{\partial p}\cdot\bar\nabla)^t{\phi}	(x, p) = 0
\end{equation}
This has a simple meaning: if $\Psi(x,P,Y)$ satisfying \eqref{FT} represents CHS field of rank $s$ and depth $t$ its weight is $w=s-t-1$ and $\Box^\ell\Phi$ has precisely the weight $t+1-s-n$ so that the LHS of CHS equation $(\Box^\ell\Psi)\big|_{Y=P^\pm=0}=0$ satisfies~\eqref{new_constraint}. 
This is known as partial conservation condition originally discussed in~\cite{Dolan:2001ih}. It can either be understood as a condition on the RHS of the CHS equation: $A_{s,t}\psi(x,p)=j(x,p)$ or as an equations satisfied by a subleading boundary value of the depth-$t$ partially-massless field in $AdS_{n+1}$~\cite{Bekaert:2013zya}.

% 
% 
% At this weight 
% This constraint corresponds to the current conservation law for depth-t FT fields. Indeed, if $\hat{\Phi}$ is depth-t FT field, then the system \eqref{components} for $A_{s, t}\hat{\Phi}$ has precisely the weight $w = t+1-s-n$. One can check that for $A_{s, t}\hat{\Phi}$ the condition \eqref{new_constraint} is satisfied automatically. On the language of parent system this was obtained in \cite{Bekaert:2013zya}.

\section{Proof of Proposition~\ref{proposition}}
\label{sec:Blt-proof}
Let us first prove the following Lemma:\\
\begin{lemma}
Let 
\begin{equation}
\label{consistencyFT}
\begin{gathered}
	 C_{\alpha_1 \alpha_2... \alpha_{\ell}}:=(I\cdot \mathcal{D} - \alpha_1 P\cdot\mathcal{D}I \cdot\dl P)\cdot ... \cdot (I\cdot \mathcal{D} - \alpha_{\ell} P\cdot\mathcal{D}I \cdot\dl P)\,,
	\end{gathered}
\end{equation}
where $\alpha_i\in\fR$ are parameters, be an operator defined on $S[s,\ell - \frac{n}{2}]$ and $\Phi \in S[s, \ell - \frac{n}{2}]$  be a lift of $\phi(x, p)$, i.e. $\Phi=L_{s,\ell - \frac{n}{2}}\phi$. Then
\begin{equation}
L^{-1}C_{\alpha_1 \alpha_2... \alpha_{\ell}}\Phi=A_{s, t}\phi
\end{equation} 
\end{lemma}
\begin{proof}
For $\alpha_1 = ... = \alpha_{\ell} = 0$ the statment os obvious as $C_{0\ldots 0}\Phi = (I\cdot\cD)^{\ell}\Phi = (-I\cdot(Y+V))^{\ell}\Box^{\ell}\Phi$ and hence coincides with $A_{s,t}\phi$
upon setting to zero $Y^A,P^\pm$.

For nonvanishing $\alpha_i$ let us show that all the terms proportional in $C_{\alpha_1 \alpha_2... \alpha_{\ell}}$ to $\alpha$ are also proportional to $P^{-}$ and hence do not contribute to $L^{-1}C_{\alpha_1 \alpha_2... \alpha_{\ell}}$. To this end observe that $I\cdot \dl P\Phi$ has the same weight as $\Phi$ and $(C_{\alpha_2... \alpha_{\ell}}\Phi(x, P, Y))$ is of weight $w = 1- \frac{n}{2}$. The following equality holds:
\begin{equation}
L^{-1} (I\cdot\mathcal{D}) C_{\alpha_2... \alpha_{\ell}}\Phi = L^{-1} C_{\alpha_1\alpha_2... \alpha_{\ell}}\Phi\,,
\end{equation}
because $\left((P\cdot\mathcal{D})(I\cdot \dl{P}) C_{\alpha_2... \alpha_{\ell}}\Phi(x, P, Y)\right)\big|_{Y=0}$ is proportional to $P^-$ thanks to the form of $\cD$ \eqref{TomasD} for $w = 1 - \half n$.

Next, observe that
\begin{equation}
	(I\cdot\mathcal{D})C_{\alpha_2... \alpha_k}\Phi =(I\cdot\mathcal{D} - \alpha_2 (P\cdot\mathcal{D})( I\cdot\dl P)) (I\cdot\mathcal{D})C_{\alpha_3... \alpha_{\ell}}\Phi
\end{equation}
because $I\cdot \cD$ commutes with every term in $C_{\alpha_1...\alpha_{\ell}}$. It follows the term containing $\alpha_2$ is proportional to $P^-$ at $Y=0$ and hence
\begin{equation}
L^{-1} (I\cdot\mathcal{D}) (I\cdot\mathcal{D}) C_{\alpha_3... \alpha_k}\Phi = L^{-1} C_{\alpha_1\alpha_2... \alpha_k}\Phi\,.
\end{equation}
Repeating the above steps ${\ell-2}$ times one finds:
\begin{equation}
\label{projectedFT}
	L^{-1}(I\cdot D)^{\ell}\Phi(x, P, Y) = L^{-1} C_{\alpha_1\alpha_2... \alpha_{\ell}}\Phi(x, P, Y)\,.
\end{equation}
\end{proof}

Let us return to the proof of Proposition \bref{proposition}. Consider $B^{\ell-1}_k\Phi(x, P, Y) \in S[s, 1 - \frac{n}{2}]$ and let $\tilde{\phi}(x, p) = L^{-1}B^{\ell-1}_k\Phi$. The analysis of Appendix \ref{sec:lift-obst} shows that $\tilde\phi=0$ is equivalent to $B^{\ell-1}_k\Phi=0$.
% that terms proportional to $P^+$ in $B^{\ell-1}_k\Phi(x, P, Y)\Big|_{Y=0}$ are proportional to $(dl{p}\cdot \bar\nabla)^k\tilde{\phi}$. It means that the equation $\tilde{\phi} = 0$ can be written as $B^{\ell-1}_k\Phi(x, P, Y)\big |_{Y=0}$ without referring to particular components.

However, that equation $B^{\ell}_k\Phi\Big|_{Y=P^\pm=0}=0$ is equivalent to $B^{\ell}_k\Phi\Big|_{Y=0}=0$ is not obvious.
% we can't easily conclude what are the terms proportional to $P^+$. It is due to the fact that for the system $S[s, -\frac{n}{2}]$ the fourth constraint $\mathcal{D}\dl P$ vanishes.
To see that it is nevertheless the case let us denote $\Psi=B^{\ell-1}_k\Phi$, $\Psi\in S[s, 1- \half n]$. The equation $B_t\Psi=0$, $t=s-2+\frac{n}{2}$ can be regarded in $n=4$ as a maximal depth CHS equations while in $n>4$ as a so-called  ``long'' CHS fields~\cite{Metsaev:2016oic}. Let us now find find terms in $B_{t}\Psi\Big|_{Y=0}$, proportional to $P^{-}$. The result is
\begin{equation}
	B_t\Phi\Big|_{Y=0} = \sum_{k=0}^{s}(\frac{k}{t} -1) (P^{-})^{k}L^{-1}(\dl {P^-})^k(\dl Y \cdot \dl Y)\Phi
\end{equation}
One can check that in addition to CHS equations the additional ones, i.e. conformal gauge conditions, appear only in dim $n = 4$ and are proportional to $(P^{-})^s$ \cite{Chekmenev:2015kzf}\cite{Bekaert:2013zya} but the coefficient vanishes in this case. In this way we conclude that $L^{-1}B_k^\ell \Phi=0$ and $(B_k^\ell \Phi)_{Y=P^\pm=0}=0$ are equivalent, giving the statement of Proposition \bref{proposition}.

\section{Partially gauge-fixed PM operators}
\label{sec:gauge_pm}

Some of the operators entering factorized FT equation are known in the literature  \cite{Skvortsov:2007kz}, \cite{Campoleoni:2012th}. Namely, the equation $\bar{B}_1\Phi = 0$, $\Phi\in S[s,s-2]$ in terms of tensor fields reads as:
\begin{multline}
	\lbrace \bar\nabla^2 +\frac{2J}{n}(-s + (n+s-3)(s-2) ) -\\-
	(p\cdot\bar\nabla)(\frac{\partial}{\partial p}\cdot\bar\nabla)+\frac{1}{n+2s-4}p^2(\frac{\partial}{\partial p}\cdot\bar\nabla)^2 \rbrace \phi(x,p)=0
\end{multline}
It is easy to check that this is precisely the Fronsdal equations in the gauge where $\phi(x,p)$ is traceless, i.e. $\dl{p}\cdot\dl{p}\phi(x,p)=0$. The residual gauge transformations are 

\begin{equation}
\delta \Phi(x, p) = p\cdot\bar\nabla \,\,\xi(x, p)\,,
\end{equation}
where $\xi$ is a subject to the following constraints: 
\begin{equation}
\dl{p}\cdot\bar\nabla\,\,\xi(x, p) = 0\,, \qquad  \dl{p}\cdot \dl{p} \,\, \xi(x, p)=0\,.
\end{equation}

In a similar fashion we may write (some of) the gauge transformations for $\bar{B}_t$:
\begin{equation}
\label{PM_n}
	\begin{gathered}
			\delta \Phi(x, p) = \lbrace (p\cdot \bar\nabla)^{t} + ...\rbrace\lambda(x, p)\,,\\
	\dl{p}\cdot\bar\nabla\,\,\lambda(x, p) = 0, \qquad 
	\dl{p}\cdot \dl{p} \,\, \lambda(x, p) =0\,.
	\end{gathered}
\end{equation}
%Using \eqref{gauge_transf} one makes sure that transformations $\Phi \sim \Phi + (P\cdot\mathcal{D})^t \lambda$, where $\lambda$ is traceless and is a subject to the constraint $ I\cdot\frac{\partial}{\partial P}\lambda = 0$ are in the kernel of $B_t$. In tensor notation these constraints are written in \eqref{PM_n}
%It can be easily seen that $(P\cdot\mathcal{D})^t \lambda$ is traceless. 
In particular, for $s=2$, $n=4$ one has:
\begin{equation}
	\begin{gathered}
		\bar{B_1}\phi(x, p) = 0,\quad \delta\phi(x, p) = p\cdot\bar\nabla\lambda_1(x, p), \quad \dl p \cdot \bar\nabla\,\,\lambda_1  = 0\,,\\
		\bar{B_2}\phi(x, p) = 0, \quad \delta\phi(x, p) = ((p\cdot\bar\nabla)^2 - \frac{1}{4}p^2\bar\nabla^2)\lambda_2(x)\,.
	\end{gathered}
\end{equation}

{\footnotesize

\addtolength{\baselineskip}{-3pt}
\addtolength{\parskip}{-3pt}

\providecommand{\href}[2]{#2}\begingroup\raggedright\endgroup

\bibliographystyle{utphys}

\begin{thebibliography}{10}

\bibitem{Fradkin:1985am}
E.~S. Fradkin and A.~A. Tseytlin, ``{C}onformal {S}upergravity,'' {\em Phys.
  Rept.} {\bf 119} (1985)
233--362.
%%CITATION = PRPLC,119,233;%%.

\bibitem{Segal:2002gd}
A.~Y. Segal, ``{C}onformal higher spin theory,'' {\em Nucl. Phys.} {\bf B664}
  (2003) 59--130,
\href{http://www.arXiv.org/abs/hep-th/0207212}{{\tt hep-th/0207212}}.
%%CITATION = HEP-TH 0207212;%%.

\bibitem{Tseytlin:2002gz}
A.~A. Tseytlin, ``{O}n limits of superstring in {A}d{S}(5) x {S}**5,'' {\em
  Theor. Math. Phys.} {\bf 133} (2002) 1376--1389,
\href{http://www.arXiv.org/abs/hep-th/0201112}{{\tt hep-th/0201112}}.
%%CITATION = HEP-TH/0201112;%%.

\bibitem{Bekaert:2010ky}
X.~Bekaert, E.~Joung, and J.~Mourad, ``{Effective action in a higher-spin
  background},'' {\em JHEP} {\bf 02} (2011) 048,
\href{http://www.arXiv.org/abs/1012.2103}{{\tt 1012.2103}}.
%%CITATION = ARXIV:1012.2103;%%.

\bibitem{Bonezzi:2017mwr}
R.~Bonezzi, ``{Induced Action for Conformal Higher Spins from Worldline Path
  Integrals},'' {\em Universe} {\bf 3} (2017), no.~3, 64,
\href{http://www.arXiv.org/abs/1709.00850}{{\tt 1709.00850}}.
%%CITATION = ARXIV:1709.00850;%%.

\bibitem{Metsaev:2008fs}
R.~Metsaev, ``{S}hadows, currents and {A}d{S},'' {\em Phys.Rev.} {\bf D78}
  (2008) 106010,
\href{http://www.arXiv.org/abs/0805.3472}{{\tt 0805.3472}}.
%%CITATION = ARXIV:0805.3472;%%.

\bibitem{Metsaev:2009ym}
R.~R. Metsaev, ``{G}auge invariant two-point vertices of shadow fields,
  {A}d{S}/{CFT}, and conformal fields,'' {\em Phys. Rev.} {\bf D81} (2010)
  106002,
\href{http://www.arXiv.org/abs/0907.4678}{{\tt 0907.4678}}.
%%CITATION = 0907.4678;%%.

\bibitem{Bekaert:2012vt}
X.~Bekaert and M.~Grigoriev, ``{N}otes on the ambient approach to boundary
  values of {A}d{S} gauge fields,'' {\em J.Phys.} {\bf A46} (2013) 214008,
\href{http://www.arXiv.org/abs/1207.3439}{{\tt 1207.3439}}.
%%CITATION = ARXIV:1207.3439;%%.

\bibitem{Bekaert:2013zya}
X.~Bekaert and M.~Grigoriev, ``{H}igher order singletons, partially massless
  fields and their boundary values in the ambient approach,'' {\em Nucl.Phys.}
  {\bf B876} (2013) 667--714,
\href{http://www.arXiv.org/abs/1305.0162}{{\tt 1305.0162}}.
%%CITATION = ARXIV:1305.0162;%%.

\bibitem{Tseytlin:2013jya}
A.~Tseytlin, ``{O}n partition function and {W}eyl anomaly of conformal higher
  spin fields,'' {\em Nucl.Phys.} {\bf B877} (2013) 598--631,
\href{http://www.arXiv.org/abs/1309.0785}{{\tt 1309.0785}}.
%%CITATION = ARXIV:1309.0785;%%.

\bibitem{Metsaev:2007fq}
R.~Metsaev, ``{O}rdinary-derivative formulation of conformal low spin fields,''
  {\em JHEP} {\bf 1201} (2012) 064,
\href{http://www.arXiv.org/abs/0707.4437}{{\tt 0707.4437}}.
%%CITATION = ARXIV:0707.4437;%%.

\bibitem{Metsaev:2007rw}
R.~R. Metsaev, ``{O}rdinary-derivative formulation of conformal totally
  symmetric arbitrary spin bosonic fields,'' {\em JHEP} {\bf 06} (2012) 062,
\href{http://www.arXiv.org/abs/0709.4392}{{\tt 0709.4392}}.
%%CITATION = ARXIV:0709.4392;%%.

\bibitem{Joung:2012qy}
E.~Joung and K.~Mkrtchyan, ``{A} note on higher-derivative actions for free
  higher-spin fields,'' {\em JHEP} {\bf 1211} (2012) 153,
\href{http://www.arXiv.org/abs/1209.4864}{{\tt 1209.4864}}.
%%CITATION = ARXIV:1209.4864;%%.

\bibitem{Deser:1983mm}
S.~Deser and R.~I. Nepomechie, ``{Gauge Invariance Versus Masslessness in De
  Sitter Space},'' {\em Annals Phys.} {\bf 154} (1984)
396.
%%CITATION = APNYA,154,396;%%.

\bibitem{Deser:2001pe}
S.~Deser and A.~Waldron, ``{G}auge invariances and phases of massive higher
  spins in ({A})d{S},'' {\em Phys.Rev.Lett.} {\bf 87} (2001) 031601,
\href{http://www.arXiv.org/abs/hep-th/0102166}{{\tt hep-th/0102166}}.
%%CITATION = HEP-TH/0102166;%%.

\bibitem{Metsaev:2014iwa}
R.~Metsaev, ``{A}rbitrary spin conformal fields in ({A})d{S},''
\href{http://www.arXiv.org/abs/1404.3712}{{\tt 1404.3712}}.
%%CITATION = ARXIV:1404.3712;%%.

\bibitem{Beccaria:2015vaa}
M.~Beccaria and A.~A. Tseytlin, ``{O}n higher spin partition functions,'' {\em
  J. Phys.} {\bf A48} (2015), no.~27, 275401,
\href{http://www.arXiv.org/abs/1503.08143}{{\tt 1503.08143}}.
%%CITATION = ARXIV:1503.08143;%%.

\bibitem{Beccaria:2016tqy}
M.~Beccaria and A.~A. Tseytlin, ``{Iterating free-field AdS/CFT: higher spin
  partition function relations},'' {\em J. Phys.} {\bf A49} (2016), no.~29,
  295401,
\href{http://www.arXiv.org/abs/1602.00948}{{\tt 1602.00948}}.
%%CITATION = ARXIV:1602.00948;%%.

\bibitem{Nutma:2014pua}
T.~Nutma and M.~Taronna, ``{O}n conformal higher spin wave operators,'' {\em
  JHEP} {\bf 1406} (2014) 066,
\href{http://www.arXiv.org/abs/1404.7452}{{\tt 1404.7452}}.
%%CITATION = ARXIV:1404.7452;%%.

\bibitem{Paneitz:1983}
S.~Paneitz, ``{A} quartic conformally covariant differential operator for
  arbitrary pseudo- riemannian manifolds,'' {\em Summary appeared in SIGMA 4
  (2008) 036, arXiv:0803.4331.} (1983).

\bibitem{Fradkin:1981jc}
E.~S. Fradkin and A.~A. Tseytlin, ``{One Loop Beta Function in Conformal
  Supergravities},'' {\em Nucl. Phys.} {\bf B203} (1982)
157--178.
%%CITATION = NUPHA,B203,157;%%.

\bibitem{GJMS}
C.~R. Graham, R.~Jenne, L.~J. Mason, and G.~A.~J. Sparling, ``Conformally
  invariant powers of the laplacian, i: Existence,'' {\em Journal of the London
  Mathematical Society} {\bf s2-46} (1992), no.~3, 557--565.

\bibitem{Gover:2005mn}
A.~R. Gover, ``{Laplacian operators and Q-curvature on conformally Einstein
  manifolds},''
\href{http://www.arXiv.org/abs/math/0506037}{{\tt math/0506037}}.
%%CITATION = MATH/0506037;%%.

\bibitem{Eastwood}
M.~Eastwood, ``{N}otes on conformal differential geometry,'' {\em Suppl. Rendi.
  Circ. Mat. Palermo} {\bf 43} (1996) 57.

\bibitem{BEG}
T.~Bailey, M.~Eastwood, and A.~Gover, ``\,{T}homas's structure bundle for
  conformal, projective and related structures,'' {\em Rocky Mountain J. Math.}
  {\bf 24} (1994) 1191--1217.

\bibitem{Cap:2002aj}
A.~\v{C}ap and A.~R. Gover, ``{S}tandard {T}ractors and the {C}onformal
  {A}mbient {M}etric {C}onstruction,'' {\em Annals Global Anal. Geom.} {\bf 24}
  (2003) 231--259,
\href{http://www.arXiv.org/abs/math/0207016}{{\tt math/0207016}}.
%%CITATION = MATH/0207016;%%.

\bibitem{Barnich:2006pc}
G.~Barnich and M.~Grigoriev, ``{P}arent form for higher spin fields on anti-de
  {S}itter space,'' {\em JHEP} {\bf 08} (2006) 013,
\href{http://www.arXiv.org/abs/hep-th/0602166}{{\tt hep-th/0602166}}.
%%CITATION = HEP-TH/0602166;%%.

\bibitem{Bekaert:2009fg}
X.~Bekaert and M.~Grigoriev, ``{M}anifestly conformal descriptions and higher
  symmetries of bosonic singletons,'' {\em SIGMA} {\bf 6} (2010) 038,
\href{http://www.arXiv.org/abs/0907.3195}{{\tt 0907.3195}}.
%%CITATION = ARXIV:0907.3195;%%.

\bibitem{Grigoriev:2011gp}
M.~Grigoriev and A.~Waldron, ``{M}assive {H}igher {S}pins from {BRST} and
  {T}ractors,'' {\em Nucl. Phys.} {\bf B853} (2011) 291--326,
\href{http://www.arXiv.org/abs/1104.4994}{{\tt 1104.4994}}.
%%CITATION = 1104.4994;%%.

\bibitem{Gover:2008sw}
A.~R. Gover, A.~Shaukat, and A.~Waldron, ``{T}ractors, {M}ass and {W}eyl
  {I}nvariance,'' {\em Nucl. Phys.} {\bf B812} (2009) 424--455,
\href{http://www.arXiv.org/abs/0810.2867}{{\tt 0810.2867}}.
%%CITATION = 0810.2867;%%.

\bibitem{Dirac:1936fq}
P.~A.~M. Dirac, ``{W}ave equations in conformal space,'' {\em Annals Math.}
  {\bf 37} (1936)
429--442.
%%CITATION = ANMAA,37,429;%%.

\bibitem{Gover:2002ay}
A.~R. Gover and L.~J. Peterson, ``{C}onformally invariant powers of the
  {L}aplacian, {Q}-curvature, and tractor calculus,'' {\em Commun. Math. Phys.}
  {\bf 235} (2002)
\href{http://www.arXiv.org/abs/math-ph/0201030}{{\tt math-ph/0201030}}.
%%CITATION = MATH-PH/0201030;%%.

\bibitem{FG}
C.~Fefferman and C.~Graham, ``{C}onformal {I}nvariants,'' {\em Ast\'erisque,
  Numero Hors Serie} (1985) 95--116.

\bibitem{LeBrun85}
C.~R. LeBrun, ``Ambi-twistors and einstein's equations,'' {\em Classical and
  Quantum Gravity} {\bf 2} (1985), no.~4, 555.

\bibitem{Drew:1980yk}
M.~S. Drew and J.~D. Gegenberg, ``{Conformally covariant massless spin-2 field
  equations},'' {\em Nuovo Cim.} {\bf A60} (1980)
41--56.
%%CITATION = NUCIA,A60,41;%%.

\bibitem{Barut:1982nj}
A.~O. Barut and B.-W. Xu, ``{On conformally covariant spin-2 and spin 3/2
  equations},'' {\em J. Phys.} {\bf A15} (1982)
L207--L210.
%%CITATION = JPAGA,A15,L207;%%.

\bibitem{Erdmenger:1997wy}
J.~Erdmenger and H.~Osborn, ``{C}onformally covariant differential operators:
  {S}ymmetric tensor fields,'' {\em Class.Quant.Grav.} {\bf 15} (1998)
  273--280,
\href{http://www.arXiv.org/abs/gr-qc/9708040}{{\tt gr-qc/9708040}}.
%%CITATION = GR-QC/9708040;%%.

\bibitem{Vasiliev:2009ck}
M.~Vasiliev, ``{B}osonic conformal higher-spin fields of any symmetry,'' {\em
  Nucl.Phys.} {\bf B829} (2010) 176--224,
\href{http://www.arXiv.org/abs/0909.5226}{{\tt 0909.5226}}.
%%CITATION = ARXIV:0909.5226;%%.

\bibitem{Chekmenev:2015kzf}
A.~Chekmenev and M.~Grigoriev, ``{Boundary values of mixed-symmetry massless
  fields in AdS space},'' {\em Nucl. Phys.} {\bf B913} (2016) 769--791,
\href{http://www.arXiv.org/abs/1512.06443}{{\tt 1512.06443}}.
%%CITATION = ARXIV:1512.06443;%%.

\bibitem{Alkalaev:2011zv}
K.~Alkalaev and M.~Grigoriev, ``{U}nified {BRST} approach to (partially)
  massless and massive {A}d{S} fields of arbitrary symmetry type,'' {\em Nucl.
  Phys.} {\bf B853} (2011) 663--687,
\href{http://www.arXiv.org/abs/1105.6111}{{\tt 1105.6111}}.
%%CITATION = 1105.6111;%%.

\bibitem{Alkalaev:2009vm}
K.~B. Alkalaev and M.~Grigoriev, ``{U}nified {BRST} description of {A}d{S}
  gauge fields,'' {\em Nucl. Phys.} {\bf B835} (2010) 197--220,
\href{http://www.arXiv.org/abs/0910.2690}{{\tt 0910.2690}}.
%%CITATION = 0910.2690;%%.

\bibitem{Gover:2008pt}
A.~R. Gover, A.~Shaukat, and A.~Waldron, ``{W}eyl {I}nvariance and the
  {O}rigins of {M}ass,'' {\em Phys. Lett.} {\bf B675} (2009) 93--97,
\href{http://www.arXiv.org/abs/0812.3364}{{\tt 0812.3364}}.
%%CITATION = 0812.3364;%%.

\bibitem{Eastwood:1985eh}
M.~G. Eastwood and M.~Singer, ``{A conformally invariant Maxwell gauge},'' {\em
  Phys. Lett.} {\bf A107} (1985)
73--74.
%%CITATION = PHLTA,A107,73;%%.

\bibitem{Dolan:2001ih}
L.~Dolan, C.~R. Nappi, and E.~Witten, ``{C}onformal operators for partially
  massless states,'' {\em JHEP} {\bf 0110} (2001) 016,
\href{http://www.arXiv.org/abs/hep-th/0109096}{{\tt hep-th/0109096}}.
%%CITATION = HEP-TH/0109096;%%.

\bibitem{Metsaev:2016oic}
R.~R. Metsaev, ``{Long, partial-short, and special conformal fields},'' {\em
  JHEP} {\bf 05} (2016) 096,
\href{http://www.arXiv.org/abs/1604.02091}{{\tt 1604.02091}}.
%%CITATION = ARXIV:1604.02091;%%.

\bibitem{Skvortsov:2007kz}
E.~D. Skvortsov and M.~A. Vasiliev, ``{T}ransverse invariant higher spin
  fields,'' {\em Phys. Lett.} {\bf B664} (2008) 301--306,
\href{http://www.arXiv.org/abs/hep-th/0701278}{{\tt hep-th/0701278}}.
%%CITATION = HEP-TH/0701278;%%.

\bibitem{Campoleoni:2012th}
A.~Campoleoni and D.~Francia, ``{M}axwell-like {L}agrangians for higher
  spins,'' {\em JHEP} {\bf 1303} (2013) 168,
\href{http://www.arXiv.org/abs/1206.5877}{{\tt 1206.5877}}.
%%CITATION = ARXIV:1206.5877;%%.

\end{thebibliography}
}
\end{document}